\documentclass [journal,twocolumn,10pt]{IEEEtran}
\usepackage{amssymb,amsmath,latexsym}
\usepackage{amsthm}
\usepackage{enumitem}
\usepackage{caption}

\newcommand{\N}{\mathbb N}

\newcommand{\F}{\mathbb F}

\DeclareMathOperator{\Tr}{Tr}

\DeclareMathOperator{\im}{im}
\newtheorem{theorem}{Theorem}
\newtheorem{definition}{Definition}
\newtheorem{lemma}{Lemma}
\newtheorem{corollary}{Corollary}
\newtheorem{proposition}{Proposition}
\newtheorem{remark}{Remark}

\newtheorem{conjecture}{Conjecture}

\DeclareMathOperator{\Ker}{ker}

\begin{document}
\title{On CCZ-equivalence of the inverse function}
	\author{{Lukas K\"olsch}
	\thanks{\noindent Lukas K\"olsch is with Department of Mathematics, University of Rostock, Germany.
		\textbf{Email}: lukas.koelsch@uni-rostock.de
	}
}
\maketitle
\begin{abstract}
The inverse function $x \mapsto x^{-1}$ on $\F_{2^n}$ is one of the most studied functions in cryptography due to its widespread use as an S-box in block ciphers like AES. In this paper, we show that, if $n\geq 5$, every function that is CCZ-equivalent to the inverse function is already EA-equivalent to it. This confirms a conjecture by Budaghyan, Calderini and Villa. We also prove that every permutation that is CCZ-equivalent to the inverse function is already affine equivalent to it. The majority of the paper is devoted to proving that there is no permutation polynomial of the form $L_1(x^{-1})+L_2(x)$ over $\F_{2^n}$ if $n\geq 5$, where $L_1,L_2$ are nonzero linear functions. In the proof, we combine Kloosterman sums, quadratic forms and tools from additive combinatorics.
\end{abstract}

\section{Introduction}

Vectorial Boolean functions play a big role in the design of symmetric cryptosystems as building blocks of block ciphers~\cite{carletbuch}. To resist differential attacks, a vectorial Boolean function should have low differential uniformity.
\begin{definition}
	A function $F \colon \F_{2^n} \rightarrow \F_{2^n}$ has differential uniformity $d$, if 
	\begin{equation*}
		d=\max_{a\in \F_{2^n}^*,b\in \F_{2^n}} |\{x \colon F(x)+F(x+a)=b\}|.
	\end{equation*}
	A function with differential uniformity $2$ is called almost perfect nonlinear (APN) on $\F_{2^n}$.
\end{definition}
As the differential uniformity is always even, APN functions yield the best resistance against differential attacks. 

The differential uniformity and in particular the APN property is preserved by certain transformations, which define equivalence relations on the set of vectorial Boolean functions. The most widely used notions of equivalence are \emph{affine equivalence}, \emph{extended affine equivalence} and \emph{CCZ-equivalence}. 

\begin{definition} \label{def:equiv}

	Two functions $F_1,F_2 \colon \F_{2^n}\rightarrow \F_{2^n}$ are called \emph{extended affine equivalent (EA-equivalent)} if there are affine permutations $A_1,A_2$ and an affine mapping $A_3$ from $\F_{2^n}$ to itself such that 
\begin{equation}
A_1(F_1(A_2(x)))+A_3(x)=F_2(x).
\label{eq:ea}
\end{equation}
$F_1$ and $F_2$ are called \emph{affine equivalent} if they are EA-equivalent and it is possible to choose $A_3=0$ in Eq.~\eqref{eq:ea}. 

	Moreover, $F_1$ and $F_2$ are called \emph{CCZ-equivalent} if there is an affine, bijective mapping $A$ on $\F_{2^n} \times \F_{2^n}$ that maps the graph of $F_1$, denoted by $G_{F_1} = \{(x,F_1(x)) \colon x \in \F_{2^n}\}$, to the graph of $F_2$.
\end{definition}

It is obvious that two functions that are affine equivalent are also EA-equivalent. Furthermore, two EA-equivalent functions are also CCZ-equivalent. In general, the concepts of CCZ-equivalence and EA equivalence do differ, for example a bijective function is always CCZ-equivalent to its compositional inverse, which is not the case for EA-equivalence. Note also that the size of the image set is invariant under affine equivalence, which is generally not the case for the other two more general notions. 

Particularly well studied are APN monomials, a list of all known infinite families is given in Table~\ref{t:APN}. The table is generally believed to be complete. We want to note that, while the inverse function is not APN in even dimension, its differential uniformity is in this case $4$, which is the lowest known differential uniformity for a bijection known as of writing this paper in all even dimensions larger than $6$.

\begin{table}[ht]
\captionsetup{font=scriptsize}
	\centering
	\begin{tabular}{ ||c|c|c|c|c|| } 
	 \hline
		& Exponent & Conditions & \\
	 \hline  \hline
	 Gold & $2^r+1$ & $\gcd(r,n)=1$   & \cite{TIT:Gold68}\\ 
	 \hline
	 Kasami & $2^{2r}-2^r+1$ & $\gcd(r,n)=1$ &\cite{IC:Kasami} \\ 
	 \hline
	 Welch & $2^t+3$ &$n=2t+1$ & \cite{TIT:Dobbertin99} \\ 
	 \hline
	 Niho & $2^t-2^{\frac{t}{2}}-1$ & $n=2t+1$, $t$ even &  \cite{IC:Dobbertin99}\\ 
	 & $2^t-2^{\frac{3t+1}{2}}-1$ & $n=2t+1$, $t$ odd  & \\ 
	 \hline
	 Inverse & $2^n-2$ & $n$ odd &   \cite{EC:Nyb}\\ 
	 \hline
	 Dobbertin & $2^{4r}+2^{3r}+2^{2r}+2^r-1$ & $n=5r$  &   \cite{FFA:Dobbertin}\\ [1ex] 
	 \hline
	\end{tabular}
	\caption{List of known APN exponents over $\F_{2^n}$ up to CCZ-equivalence. }
	\label{t:APN}
\end{table}

Especially the concept of CCZ-equivalence has been extensively studied in the past years since it is a powerful tool for constructing and 
studying cryptological functions.

Let us denote by the \emph{EA-class} of $F$ the set of all functions EA-equivalent to a vectorial Boolean function $F$, and similarly by \emph{CCZ-class} of $F$ the set of all functions CCZ-equivalent to $F$. Since EA-equivalence is a special case of CCZ-equivalence, we can partition a CCZ-class into EA-classes. Experimental results show that for many functions $F \colon \F_{2^n} \rightarrow \F_{2^n}$ the CCZ-class of $F$ coincides with its EA-class if $F$ is not a permutation. If $F$ is a permutation, then its CCZ class often consists of precisely $2$ EA-classes, with $F$ and $F^{-1}$ being representatives for the two EA-classes. Of course, it is possible to have more than two EA-classes in one CCZ-class, this is for example the case for Gold functions in odd dimension \cite[Theorem 1]{eaccz}. Additionally, it is possible that a permutation and its inverse are in the same EA-class, this happens for example naturally for involutions like the inverse function $x \mapsto x^{2^n-2}$ over $\F_{2^n}$.

Since CCZ-equivalence is a much more difficult concept than EA-equivalence, it is desirable to understand when the CCZ-class of $F$ can be described entirely by combining EA-equivalence and the inverse transformation (in the case that $F$ is invertible). This is a problem that is open even for most of the APN monomials in Table~\ref{t:APN}. As noted earlier, the CCZ-class of Gold functions in odd dimension contains more than $2$ EA-classes and thus cannot be described with EA-equivalence and the inverse transformation alone. For all other monomials, this question is still open. Budaghyan, Calderini and Villa conjectured the following based on a computer search for small values of $n$:
 
\begin{conjecture}[{\cite[Conjecture 4.14]{conj}}] \label{conj:bud}
	Let $F(x)=x^d$ be a non-Gold APN monomial or the inverse function over $\F_{2^n}$. Then, every function that is CCZ-equivalent to $F$ is EA-equivalent to $F$ or to $F^{-1}$ (if it exists).
\end{conjecture}

In this paper we are going to confirm this conjecture in the case of the inverse function, i.e. we show that the CCZ-class of the function $x \mapsto x^{2^n-2}$ on $\F_{2^n}$ coincides with its EA-class. This is to our knowledge the first theoretical result of this kind. Note that, in many ways, the inverse function is actually the most interesting case because of its widespread use in cryptography, most famously as the S-box in AES.\\

In~\cite{yoshiara,dempwolff} it was shown that two power functions $x \mapsto x^k$ and $x \mapsto x^l$ on $\F_{2^n}$ are CCZ-equivalent if and only if $k \equiv 2^il \pmod{2^n-1}$ or $kl \equiv 2^i \pmod{2^n-1}$ for some $i \in \N$. Our result can then be seen as a generalization of this result for the specific case $k=2^n-2$ since the only power functions that are EA-equivalent to the inverse function are the power functions with $l \equiv 2^i(2^n-2) \pmod{2^n-1}$. Our approach in this paper is very different from the one used in~\cite{yoshiara,dempwolff}, which relies on sophisticated  tools exploiting the structure of the automorphism groups of power functions. It seems difficult to use these group theoretical tools to settle the complete relationship of EA- and CCZ-equivalence between (APN) power functions and arbitrary functions because the different power functions display different behavior. Indeed, as mentioned above, there exist functions that are CCZ-equivalent to an APN Gold function but not EA-equivalent to any power function. Computer results in low dimensions show that for most other power functions, such functions do not exist (and a proof of this nonexistence in the case of the inverse function is of course the subject of this paper). The general approach we present in this paper (see Proposition~\ref{prop:start}) is quite flexible and may be used also for different functions, even non-power functions. However, the precise proof  depends a lot on the specific function, and our proof is to a large degree tailored specifically towards the inverse function. \\

Because of the importance of permutations for the design of block ciphers, it is also interesting to search for permutations inside the CCZ-class of a function $F$. Indeed, a way to find permutations with good cryptographic properties is to look for a permutation in the CCZ-class of a non-permutation with good cryptographic properties, since CCZ-equivalence preserves many important properties. This is precisely the technique that was used to find the only known APN permutation in even dimension \cite{dillonperm}. Thus it is a very interesting question to classify all permutations that are in the CCZ-class of an APN function. Treating this problem for infinite families is however very difficult. To the authors knowledge, the only result in this direction is found in a recent paper~\cite{goldkasami}, where it was shown that there is no permutation in the CCZ-class of the APN Gold functions over $\F_{2^n}$ with $n$ even and that there is no permutation in the CCZ-class of APN Kasami functions over $\F_{2^n}$ if $n$ is divisible by $4$. The proof technique used in~\cite{goldkasami} relies on a careful analysis of the bent component functions of the Gold and Kasami functions. 

In this paper, we will also classify all permutations in the CCZ-class of the inverse function (both in odd and even dimension). We show that (excluding some sporadic cases in low dimension) the only permutations in the CCZ-class are the ``trivial'' ones, i.e. the ones that are affine equivalent to the inverse function. 
Of course, since the inverse function does not have any bent components, the approach is necessarily different from the approach in~\cite{goldkasami}. Instead, we are going to use the following proposition which connects both of the problems we mentioned to a special type of permutation polynomial. The proof of the proposition can be found in~\cite{waifipaper}, co-written by the author of this paper, that deals with the same questions as this paper but only achieved partial results. 

\begin{proposition}[\cite{waifipaper}] \label{prop:start} 
\begin{description}
	\item[(a)] Let $F \colon \F_{2^n} \rightarrow \F_{2^n}$ and assume no permutation of the form $F(x)+L(x)$ exists where $L$ is a non-zero linear function. Then every permutation that is EA-equivalent to $F$ is already affine equivalent to it. 
		\item[(b)] Let $F \colon \F_{2^n} \rightarrow \F_{2^n}$ and assume no permutation of 
		the form $L_1(F(x))+L_2(x)$ exists where $L_1,L_2$ are non-zero linear functions. 
		Then every function that is CCZ-equivalent to $F$ is EA-equivalent to 
		$F$ or $F^{-1}$ (if it exists). Moreover, all permutations that are CCZ-equivalent to $F$ are affine equivalent to $F$ or $F^{-1}$.
	\end{description}
	
\end{proposition}

Based on this proposition, we are going to study permutation polynomials of the form $F(x) = L_1(x^{2^n-2})+L_2(x)$ where $L_1,L_2$ are linear functions. For simplicity, we will use for the rest of this paper the usual convention $0^{-1}= 0$ which allows us to denote the inverse function by $x \mapsto x^{-1}$.

\section{Permutation polynomials of the form  $L_1(x^{-1})+L_2(x)$}

\subsection{Preliminaries and preparation}

Let us first introduce some basic concepts and notation. 

We denote by $\Tr$ the absolute trace function mapping $\F_{2^n}$ to $\F_2$. The value of $n$ is here always taken from the context. We then define the hyperplanes of $\F_{2^n}$ as $H_a=\{x \in \F_{2^n} \colon \Tr(ax)=0\}$ for  $a \in \F_{2^n}^*$. 

For a set $A \subseteq \F_{2^n}$, we denote by $1/A$ the set of all inverses of $A$, i.e. $1/A=\{1/a \colon a \in A\setminus\{0\}\}$. Further, we define the product set $A\cdot A = \{a_1a_2 \mid a_1,a_2 \in A\}$ and $\sqrt{A}=\{\sqrt{a} \mid a \in A\}$. Note that since we are working in fields of characteristic $2$, the function $x \mapsto x^2$ is bijective, so $|\sqrt{A}|=|A|$.

For a linear mapping $L$, we denote by $L^*$ its adjoint mapping with respect to the bilinear form 
\begin{equation*}
\langle x,y \rangle=\Tr(xy),
\end{equation*}
i.e. we have 
\begin{equation*}
\Tr(L(x)y)=\Tr(xL^*(y))
\end{equation*}
for all $x,y \in \F_{2^n}$. For a given linear mapping $L=\sum_{i=0}^{n-1}c_ix^{2^i}$ we can explicitly describe its adjoint mapping by $L^*=\sum_{i=0}^{n-1}c_{i}^{2^{n-i}}x^{2^{n-i}}$. The following lemma is well-known, for a proof see e.g.~\cite{waifipaper}. 

\begin{lemma} \label{lem:adjoint}
	Let $L \colon \F_{2^n} \rightarrow \F_{2^n}$ be a linear mapping and $L^*$ its adjoint mapping. Then $\dim(\im(L^*))=\dim(\im(L))$ and $\dim(\ker(L^*))=\dim(\ker(L))$.
\end{lemma}

\begin{definition}
	For $a \in \F_{2^n}$, the Kloosterman sum of $a$ over $\F_{2^n}$ is defined as
	\begin{equation*}
		K_n(a)=\sum_{x \in \F_{2^n}}(-1)^{\Tr(x^{-1}+ax)}.
	\end{equation*}
\end{definition}
 
Using Kloosterman sums and the adjoint polynomial, we can give a criterion when $L_1(x^{-1})+L_2(x)$  is a permutation.
\begin{proposition}[\cite{waifipaper}] \label{prop:inverse}
	Let $L_1,L_2$ be linear mappings over $\F_{2^n}$. Then $L_1(x^{-1})+L_2(x)$ is a permutation of $\F_{2^n}$ if and only if
		\begin{equation*}
		K_n(L_1^*(b)L_2^*(b)) =0
	\end{equation*}
	for all $b \in \F_{2^n}$ and $\ker(L_1^*)\cap \ker(L_2^*)=\{0\}$.
\end{proposition}

Proposition~\ref{prop:inverse} motivates us to investigate elements $a \in \F_{2^n}^*$ with the property that $K_n(a)=0$. We call such elements \emph{Kloosterman zeros}. 
Kloosterman zeros are used  for the construction of bent and hyperbent functions (see for example  \cite{dillonthesis,hyperbent}). Few results about the distribution of Kloosterman zeros are known. It is known that Kloosterman zeros exist for all values of $n$ (note that this is not true in characteristic $ \geq 5$  \cite{finns}). There is a way to compute the number of Kloosterman zeros \cite{goppa}, which relies on determining the class number of binary quadratic forms. However, it is difficult to use this method to derive a theoretical result on the number and distribution of Kloosterman sums. Moreover, it is known that for $n>4$, Kloosterman zeros are never contained in proper subfields of $\F_{2^n}$ \cite{zeros}.

Instead of working directly with Kloosterman sums, we are going to use dyadic approximations, which give necessary conditions for elements to be a Kloosterman zero.

Let $Q \colon \F_{2^n} \rightarrow \F_2$ be the quadratic form defined by 
\begin{equation*}
Q(x)=\sum_{0 \leq i <j <n}x^{2^i+2^j}
\end{equation*}

 for all $x \in \F_{2^n}$. 

\begin{theorem}[\cite{faruk}] \label{thm:faruk}
	Let $n\geq 4$. Then $K_{n}(a) \equiv 0 \pmod {16}$ if and only if $\Tr(a)=0$ and $Q(a)=0$.
\end{theorem}
Applied to Proposition~\ref{prop:inverse} this theorem gives a necessary condition for $L_1(x^{-1})+L_2(x)$ to be a permutation. 
\begin{corollary}\label{cor:faruk}
		If $L_1(x^{-1})+L_2(x)$ is a permutation of $\F_{2^n}$ with $n\geq 4$, then $\Tr(L_1^*(a)L_2^*(a))=Q(L_1^*(a)L_2^*(a))=0$ for all $a \in \F_{2^n}$ and $\Ker(L_1^*)\cap \Ker(L_2^*)=\{0\}$. 
\end{corollary}
 Let $B(x,y)=Q(x)+Q(y)+Q(x+y)$ be the bilinear form associated to $Q$. 

We can determine $B$ explicitly with a simple calculation.
	We have 
	\begin{align*}
			B(x,y)=&\sum_{0 \leq i <j <n}x^{2^i+2^j}+\sum_{0 \leq i <j <n}y^{2^i+2^j} \\
			&+\sum_{0 \leq i <j <n}(x+y)^{2^i+2^j} \\
			=&\sum_{i\neq j}x^{2^i}y^{2^j}=\sum_ix^{2^i}\sum_{j \neq i }y^{2^j}\\
			=&\sum_ix^{2^i}(\Tr(y)+y^{2^i}) \\
			=&\sum_i(xy)^{2^i}+\Tr(y)\sum_ix^{2^i}\\
			=&\Tr(xy)+\Tr(x)\Tr(y).
	\end{align*}

	A partial result in the classification of permutations of the form $L_1(x^{-1})+L_2(x)$ was found in~\cite{eainverse}. We want to remark that the techniques we develop in this paper to tackle the more general question differ considerably from the ones employed in~\cite{eainverse}.
\begin{theorem}[\cite{eainverse}] \label{thm:chin}
	Let $F\colon \F_{2^n} \rightarrow \F_{2^n}$ defined by  $F(x)= x^{-1}+L(x)$ with some nonzero linear mapping $L(x)$. If $n\geq 5$ then $F$ is not a permutation.
\end{theorem} 
From this theorem we can derive a straightforward corollary.
\begin{corollary}\label{cor:chin-gen}
	Let $n\geq 5$ and $F\colon \F_{2^n} \rightarrow \F_{2^n}$ be defined by  $F(x) = L_1(x^{-1})+L_2(x)$,
	where $L_1, L_2$ are non-zero linear mappings. If $L_1$ or $L_2$ is bijective, then $F$ is not a permutation on $\F_{2^n}$.
\end{corollary} 
	\begin{proof}
	Assume $F$ is a permutation and $L_1$ is bijective. Then $L_1^{-1} \circ F = x^{-1}+L_1^{-1}(L_2(x))$ is a permutation, contradicting Theorem~\ref{thm:chin}. 
	
	Since $F$ is a permutation if and only if $F(x^{-1})=L_1(x)+L_2(x^{-1})$ is a permutation, the same holds also for $L_2$.
	\end{proof}

\subsection{The proof of the non-existence of permutations of the form $L_1(x^{-1})+L_2(x)$ }

We now prove that no permutations of the form $L_1(x^{-1})+L_2(x)$ with $L_1,L_2\neq0$ exist if $n \geq 5$. The conditions given in Corollary~\ref{cor:faruk} are our starting point. Our proof consists of three steps: 
\begin{itemize}
	\item We show that if $L_1(x^{-1})+L_2(x)$ permutes $\F_{2^n}$, then the kernels of $L_1$ and $L_2$ must be translates of a subfield of $\F_{2^n}$, i.e. of the form $a \F_{2^k}$ with $a \in \F_{2^n}^*$ and $k|n$. 
	\item Under this condition, we show $k=1$, i.e. that $\dim(\ker (L_1))=\dim(\ker (L_2))=1$.
	\item We show explicitly that there is no permutation with $\dim(\ker (L_1))=\dim(\ker (L_2))=1$.
\end{itemize}

The key step here is the first one. Generally, the difficulty of the problem lies in the fact that the inverse function does not preserve the additive structure given by the linear mappings. However, if $\ker(L)=a\F_{2^k}$, then $1/\ker(L) \cup \{0\}=\frac{1}{a}\F_{2^k}$, so the kernel retains its structure after inversion, which is the key for the next steps. \\

We will use the following result from additive combinatorics which characterizes the subsets in an Abelian group with doubling constant $1$.
\begin{theorem}[{\cite[Proposition 2.7]{tao_vu}}] \label{thm:doubling}
	Let $(G,\cdot)$ be an Abelian group and $A \subseteq G$ a finite subset of $G$. Then $|A\cdot A|=|A|$ if and only if $A=gH$, where $g \in G$ and $H \leq G$ is a subgroup of $G$.
\end{theorem} 

We are now ready to prove the first step we outlined above.
\begin{theorem} \label{thm:subfield}
		Let $F \colon \F_{2^n}\rightarrow \F_{2^n}$ with $n \geq 5$ be defined by $F(x)=L_1(x^{-1})+L_2(x)$, where $L_1,L_2$ are nonzero linear mappings over $\F_{2^n}$. If $F$ is a permutation, then $\ker (L_1)$ and $\ker (L_2)$ are translates of subfields of $\F_{2^n}$, i.e. they are of the form $a\F_{2^k}$ for $a\neq 0$ and $k|n$. Moreover, we have $\ker (L_1)=L_2^*(\ker(L_1^*))$ and $\ker (L_2)=L_1^*(\ker(L_2^*))$.
\end{theorem}
\begin{proof}

We define $R(x)=L_1^*(x)L_2^*(x)$. Assume that $F$ permutes $\F_{2^n}$. Then Corollary~\ref{cor:faruk} implies $Q(R(x))=0$ and $\Tr(R(x))=0$ for all $x \in \F_{2^n}$. Let $x\in \F_{2^n}$ and $y \in \ker (L_1^*)$ (recall that we can assume $\ker(L_1^*) \neq \{0\}$ by Corollary~\ref{cor:chin-gen}). Then
\begin{align}
	0=&Q(R(x+y))\notag \\
	=&Q(R(x)+L_1^*(x)L_2^*(y)+L_1^*(y)L_2^*(x)+R(y)) \notag\\
	=&Q(R(x)+L_1^*(x)L_2^*(y)) \notag \\
	=&Q(R(x))+Q(L_1^*(x)L_2^*(y))+B(R(x),L_1^*(x)L_2^*(y)) \notag\\
	=&Q(L_1^*(x)L_2^*(y))+ \Tr(R(x)L_1^*(x)L_2^*(y)),\label{eq:1}
\end{align}
where we use $R(y)=L_1^*(y)=0$, $Q(R(x))=0$ and $\Tr(R(x))=0$ throughout the computation, as well as the bilinear form $B(x,y)=\Tr(xy)+\Tr(x)\Tr(y)$.

For every $z \in \ker(L_1^*)$, we then get using Eq.~\eqref{eq:1}
\begin{align*}
	0=&Q(L_1^*(x+z)L_2^*(y))+\Tr(R(x+z)L_1^*(x+z)L_2^*(y)) \\
	 =&Q(L_1^*(x)L_2^*(y))+\Tr(R(x+z)L_1^*(x)L_2^*(y))\\
	=&Q(L_1^*(x)L_2^*(y))+\Tr(L_1^*(x)L_2^*(x+z)L_1^*(x)L_2^*(y)) \\
	=&Q(L_1^*(x)L_2^*(y))+\Tr(R(x)L_1^*(x)L_2^*(y))\\
	&+\Tr((L_1^*(x))^2L_2^*(z)L_2^*(y)). 
\end{align*}
Adding the last equation to Eq.~\eqref{eq:1} yields
\begin{equation}\label{eq:2}
\Tr((L_1^*(x))^2L_2^*(z)L_2^*(y))=0
\end{equation}
 for all $y,z \in \ker(L_1^*)$ and $x \in \F_{2^n}$.

Setting $y=z$ in Eq.~\eqref{eq:2} we have $0=\Tr(L_1^*(x)L_2^*(y))=\Tr(xL_1(L_2^*(y)))$ for all $x \in \F_{2^n}$. Consequently, $L_2^*(\ker(L_1^*)) \subseteq \ker (L_1)$. Since $\ker(L_1^*)\cap \ker(L_2^*)=\{0\}$ by Corollary~\ref{cor:faruk} and $\dim (\ker (L_1))=\dim(\ker (L_1^*))$ by Lemma~\ref{lem:adjoint}, we get $\ker (L_1)=L_2^*(\ker(L_1^*))$. 

Again using Eq.~\eqref{eq:2}, we get 
\[\Tr\left(xL_1\left(\sqrt{L_2^*(z)L_2^*(y)}\right)\right)=0 \]
for all $y,z \in \ker(L_1^*)$ and $x \in \F_{2^n}$, so $\sqrt{L_2^*(\ker(L_1^*))\cdot L_2^*(\ker(L_1^*))} \subseteq \ker (L_1)$. Now since $L_2^*(\ker(L_1^*))=\ker (L_1)$ we have $\sqrt{\ker (L_1)\cdot \ker (L_1)}=\ker (L_1)$.

This particularly implies that 
\[|(\ker (L_1)\setminus\{0\})\cdot (\ker (L_1)\setminus\{0\})|=|(\ker (L_1)\setminus\{0\})|,\] so by Theorem~\ref{thm:doubling} we get $\ker (L_1)=aH \cup \{0\}$ where $a \in \F_{2^n}^*$ and $H\leq \F_{2^n}^*$ is a subgroup of the multiplicative group of $\F_{2^n}$. Since $|\ker (L_1)|=2^k$ for some $k\in \N$, we infer $|H|=2^k-1$. In particular, $H$ is the multiplicative group of a subfield of $\F_{2^n}$, so $\ker (L_1) = a\F_{2^k}$ for some $a \neq 0 $ and $k|n$. 

Since $F(x)=L_1(x^{-1})+L_2(x)$ is a permutation if and only if $F(x^{-1})=L_1(x)+L_2(x^{-1})$ is a permutation, we get the same results also for the kernel of $L_2$.

\end{proof}

Before we start with the second step of the proof, we need some lemmas. The first one is very well-known and can be found for instance in~\cite{LN}.
\begin{lemma} \label{lem:quad}
	The quadratic equation $ax^2+bx+c=0$ over $\F_{2^n}$ with $b \neq 0$ has solutions in $\F_{2^n}$ if and only if $\Tr(ac/b^2)=0$.
\end{lemma}

\begin{proposition} \label{prop:hyper}
	Let $F \colon \F_{2^n}\rightarrow \F_{2^n}$ be defined by $F(x)=L_1(x^{-1})+L_2(x)$, where $L_1,L_2$ are nonzero linear mappings over $\F_{2^n}$.
	$F$ is a permutation polynomial if and only if $F$ has only one zero and $L_2(a) \not\in L_1(\frac{1}{H_{1/a}})$ for all $a \in \F_{2^n}^*$.
\end{proposition}
\begin{proof}
	$F$ permutes $\F_{2^n}$ if and only if $F(x)+F(x+a) \neq 0$ for all $x \in \F_{2^n}$ and $a \in \F_{2^n}^*$, i.e.
	\begin{equation} \label{eq:prop1}
		L_1(x^{-1}+(x+a)^{-1})\neq L_2(a).
	\end{equation}
	We determine the set $M_a = \{c \in \F_{2^n} \mid \exists x \in\F_{2^n} \colon  x^{-1}+(x+a)^{-1}=c\}$. We have
	\begin{equation*}
	x^{-1}+(x+a)^{-1}=c \iff c=a^{-1} \text{ or } \Tr(\frac{1}{ac})=0.
	\end{equation*}
	Indeed, $c=a^{-1}$ if we choose $x=0$ or $x=a$, and in the other cases we can multiply the equation by $x(x+a)$ which results in the quadratic equation
	\[cx^2+acx+a = 0.\]
	By Lemma~\ref{lem:quad}, this equation has a solution in $\F_{2^n}$ if and only if $c \neq 0$ and $\Tr(1/(ac))=0$. We conclude that $M_a=1/H_{\frac{1}{a}} \cup \{a^{-1}\}$, so Eq.~\eqref{eq:prop1} gives $L_2(a)\not\in L_1(M_a)$. Observe that $L_2(a)=L_1(a^{-1})$ if and only if $a$ is a zero of $F$ and the result follows.
\end{proof}

\begin{lemma} \label{lem:hyper}
	Let $a,b,c$ be three distinct elements in  $\F_{2^n}^*$. Then $H_{a} \cup H_{b} \cup H_{c} = \F_{2^n}$ if and only if $a+b = c$. In particular, if $M=r\F_{2^k}$ with $k|n$, $k>1$ and $r \in \F_{2^n}^*$, we can always find three elements $a,b,c\in M\setminus \{0\}$, such that $H_{1/a} \cup H_{1/b} \cup H_{1/c} = \F_{2^n}$.
\end{lemma}
\begin{proof}
	Clearly, all hyperplanes have $2^{n-1}$ elements and all intersections of two distinct hyperplanes have $2^{n-2}$ elements, so
	\begin{align*}
		|H_{a} \cup H_{b} \cup H_{c} |=&|H_{a}|+|H_{b}|+|H_{c}|-|H_{a} \cap H_{b}|\\&-|H_{a} \cap H_{c}|-|H_{b} \cap H_{c}|\\&+|H_{a} \cap H_{b} \cap H_{c} | \\
		=&3\cdot 2^{n-1}-3\cdot 2^{n-2}+|H_{a} \cap H_{b} \cap H_{c} |\\
		=&2^{n-1}+2^{n-2}+|H_{a} \cap H_{b} \cap H_{c} |.
	\end{align*}
	Consequently, $H_{a} \cup H_{b} \cup H_{c}=\F_{2^n}$ if and only if $|H_{a} \cap H_{b} \cap H_{c} |=2^{n-2}$, which means $H_{a} \cap H_{b}  \subseteq H_{c}$. This is equivalent to $a+b = c$.
	
	Now let $M=r\F_{2^k}$ be a translate of a subfield of $\F_{2^n}$ with $k>1$. Choose two distinct elements $a=rs_1$, $b=rs_2$ with $s_1,s_2 \in \F_{2^k}^*$, then 
	\[\frac{1}{a}+\frac{1}{b} = \frac{1}{r}\left(\frac{1}{s_1}+\frac{1}{s_2}\right).\]
	Clearly, $1/s_1+1/s_2=1/s$ is an element in $\F_{2^k}^*$. For the three elements $a,b,rs \in M$ we have $\frac{1}{a}+\frac{1}{b} = \frac{1}{rs}$, so we have  $H_{1/a} \cup H_{1/b} \cup H_{1/rs} = \F_{2^n}$.
\end{proof}

\begin{theorem} \label{thm:kerdim1}
Let $F \colon \F_{2^n}\rightarrow \F_{2^n}$ be defined by $F(x)=L_1(x^{-1})+L_2(x)$ with $n \geq 5$, where $L_1,L_2$ are nonzero linear mappings over $\F_{2^n}$. If $F$ is a permutation, then $|\ker(L_1)|=|\ker(L_2)|=2$.
\end{theorem}
\begin{proof}
	By Theorem~\ref{thm:subfield}, we know that $\ker (L_2) = a\F_{2^k}$ for some $a \in \F_{2^n}^*$ and $k|n$. We want to show that $k=1$. Assume to the contrary that $k\geq 2$. Then there exist three distinct elements $a,b,c \in \ker (L_2)$ such that $H_{1/a} \cup H_{1/b} \cup H_{1/c} = \F_{2^n}$ by Lemma~\ref{lem:hyper}. Equivalently, $\frac{1}{H_{1/a}} \cup\frac{1}{H_{1/b}}  \cup\frac{1}{H_{1/c}} = \F_{2^n}^*$. 
	
	By Proposition~\ref{prop:hyper}, we have $L_2(x) \not\in L_1(\frac{1}{H_{1/x}})$ for all $x \in \F_{2^n}^*$. For the elements $a,b,c$, this relation becomes 
	\begin{equation*}
		0 \not\in L_1(\frac{1}{H_{1/a}}) \cup L_1(\frac{1}{H_{1/b}}) \cup L_1(\frac{1}{H_{1/c}}) = L_1(\F_{2^n}^*).
	\end{equation*}
	But since $|\ker (L_1)|>1$ by Corollary~\ref{cor:chin-gen}, this is not possible. We conclude $k=1$ and $|\ker(L_2)| = 2$.
	
	Since $L_1(x^{-1})+L_2(x)$ is a permutation if and only if $L_1(x)+L_2(x^{-1})$ is a permutation, the argument works again symmetrically for the kernel of $L_1$.
\end{proof}

Basic linear algebra shows that for two linear mappings $L,L'$ with $\ker(L)=\ker(L')$, there exists always a bijective linear mapping $L''$ such that $L' = L'' \circ L$, so (using Theorem~\ref{thm:kerdim1}) we can assume without loss of generality that $L_1(x)=x^2+ax$ for some $a \neq 0$. In fact, we can even assume $a=1$ as the following argument shows:
\begin{align}
	x^{-2}+ax^{-1}+L_2(x)&=c \label{eq:eqsubst1}\\
	\iff a^{-2}x^{-2}+a^{-1}x^{-1}+a^{-2}L_2(x)&=a^{-2}c \nonumber
\end{align}
by multiplying the equation with $a^{-2}$. After a substitution $x \mapsto x/a$, we get
\begin{equation}
	x^{-2}+x^{-1}+a^{-2}L_2(x/a)=a^{-2}c.
\label{eq:eqsubst3}
\end{equation}

Eq.~\eqref{eq:eqsubst1} has one solution for every $c \in \F_{2^n}$ if and only if Eq.~\eqref{eq:eqsubst3} has one solution for each $c$. Observe that $a^{-2}L_2(x/a)$ is still a linear mapping, so we can consider without loss of generality $L_1(x)=x^2+x$. In fact, we will instead use $L_1(x)=(x^2+x)^{2^{n-1}}=x^{2^{n-1}}+x$, which is equivalent to the case $L_1(x)=x^2+x$. Indeed, if  $F(x)=x^{-2}+x^{-1}+L_2(x)$ is a permutation, then so is $F(x^{2^{n-1}})=x^{-1}+x^{-2^{n-1}}+L_2(x^{2^{n-1}})$. The reason for this transformation is that in this case $L_1^*(x)=x^2+x$, which makes the following technical calculations slightly easier. In this case we also have $\ker(L_1)=\ker(L_1^*)=\{0,1\}$. By Theorem~\ref{thm:subfield}, we know $\ker(L_1)=L_2^*(\ker(L_1^*))$, which implies $L_2^*(1)=1$. 

\begin{theorem} \label{thm:noperm}
	There is no permutation of the form $F(x)=L_1(x^{-1})+L_2(x)$ on $\F_{2^n}$ with nonzero linear mappings $L_1,L_2$ if $n \geq 5$.
\end{theorem}
\begin{proof}
	Assume that $F$ is a permutation. By the considerations above we can assume without loss of generality that $L_1^*(x)=x^2+x$ and $L_2^*(1)=1$. We set $L_2^*(x)=\sum c_ix^{2^i}$ and derive necessary conditions on the coefficients $c_i$ and show that those conditions contradict each other. As the basis we use the conditions given in Corollary~\ref{cor:faruk} as well as Eq. \eqref{eq:1} for $y = 1$. We get
	\begin{align}
		\Tr((x^2+x)L_2^*(x))&=0 \label{eq:a1}\\
		Q((x^2+x)L_2^*(x))&=0 \label{eq:a2}\\
		Q(x^2+x)+\Tr((x^4+x^2)L_2^*(x))&=0\label{eq:a3}
	\end{align}
	for all $x \in \F_{2^n}$.
	
	We start by expanding condition~\eqref{eq:a1}. We have
	\begin{align*}
		&\Tr((x^2+x)\sum_{i=0}^{n-1} c_ix^{2^i}) \\
		&= \sum_{s=0}^{n-1}\sum_{i=0}^{n-1} c_i^{2^s}x^{2^{i+s}+2^{s+1}}+\sum_{s=0}^{n-1}\sum_{i=0}^{n-1} c_i^{2^s}x^{2^{i+s}+2^{s}}\\
		 &=\sum_{s=0}^{n-1}\sum_{i=0}^{n-1} c_{i-s}^{2^s}x^{2^{i}+2^{s+1}}+\sum_{s=0}^{n-1}\sum_{i=0}^{n-1} c_{i-s}^{2^s}x^{2^{i}+2^{s}} \\
		 &= \sum_{s=0}^{n-1}\sum_{i=0}^{n-1} (c_{i-s+1}^{2^{s-1}}+c_{i-s}^{2^{s}})x^{2^i+2^s},
	\end{align*}
	where we use a transformation $i \mapsto i-s$ in the second step and then a transformation $s \mapsto s-1$ in the left double sum in the last step. Here we view the indices of the coefficients $c_i$ modulo $n$.
	 By condition~\eqref{eq:a1} this polynomial is equal to the zero polynomial. When we check the coefficient of $x^4$ (achieved by setting $i=s=1$), we get 
	\begin{equation}
		c_1=c_0^2.
	\label{eq:cond1}
	\end{equation}
	Similarly, checking the coefficient of $x^{2^r+1}$ for $1 \leq r \leq n-1$ (achieved for $i=0$, $s=r$ and $i=r$, $s=0$) we get
	\begin{equation}
		c_{-r+1}^{2^{r-1}}+c_{-r}^{2^r}+c_{r+1}^{2^{-1}}+c_r=0
	\label{eq:cond2}
	\end{equation}
	for all $1 \leq r \leq n-1$. 
	
	We now do the same procedure for condition~\eqref{eq:a3}. We have
	\begin{align*}
		Q&(x^2+x)+\Tr((x^4+x^2)L_2^*(x)) \\
		=& Q(x^2)+Q(x)+\Tr(x^3)+\Tr(x)\\
		&+\Tr((x^4+x^2)\sum_{i=0}^{n-1} c_ix^{2^i}) \\
		=& \Tr(x^3)+\Tr(x)+\sum_{s=0}^{n-1}\sum_{i=0}^{n-1} c_i^{2^s}x^{2^{i+s}+2^{s+2}}\\&+\sum_{s=0}^{n-1}\sum_{i=0}^{n-1} c_i^{2^s}x^{2^{i+s}+2^{s+1}} \\
		=&\sum_{i=0}^{n-1}x^{2^{i+1}+2^i}+\sum_{i=0}^{n-1}x^{2^i}\\&+\sum_{s=0}^{n-1} \sum_{i=0}^{n-1}(c_{i-s+2}^{2^{s-2}}+c_{i-s+1}^{2^{s-1}})x^{2^i+2^s},
	\end{align*}	
	where we use $Q(x^2)=Q(x)$ and in the last two steps again the transformations $i \mapsto i-s$ and $s \mapsto s-2$ (in the left double sum) and $s \mapsto s-1$ (in the right double sum). Checking the coefficient of $x^8$ (achieved by $i=s=2$ in the double sum) we get
	\begin{equation}
		c_2=c_1^2+1.
		\label{eq:cond3}
	\end{equation}
	For the coefficients of $x^{2^r+1}$ for $1 \leq r \leq n-1$ (achieved by  $i=r$, $s=0$ and $i=0$, $s=r$), we similarly get
	\begin{equation}
	c_{r+2}^{2^{-2}}+c_{r+1}^{2^{-1}}+c_{-r+2}^{2^{r-2}}+c_{-r+1}^{2^{r-1}}=\begin{cases}
	1, & r \in \{1,n-1\} \\
	0, & r \in \{2,\dots,n-2\},
	\end{cases}
	\label{eq:3case}
	\end{equation}
	where the additional $1$ in the cases $r\in \{1,n-1\}$ is due to the $\Tr(x^3)$ term. Substituting $r \mapsto r-1$ and squaring the equation yields
	
	\begin{equation*}
	c_{r+1}^{2^{-1}}+c_{r}+c_{-r+3}^{2^{r-2}}+c_{-r+2}^{2^{r-1}}=\begin{cases}
	1, & r \in \{0,2\} \\
	0, & r \in \{3,\dots,n-1\}.
	\end{cases}
	\end{equation*}
	Adding this equation to Eq. \eqref{eq:cond2}, we get
		\begin{equation*}
	c_{-r+1}^{2^{r-1}}+c_{-r}^{2^r}+c_{-r+3}^{2^{r-2}}+c_{-r+2}^{2^{r-1}}=\begin{cases}
	1, & r =2 \\
	0, & r \in \{3,\dots,n-1\}.
	\end{cases}
	\end{equation*}
	We simplify the equation by substituting $r\mapsto -r$ and taking the resulting equation to the power $2^{r+2}$:
		\begin{equation}
	c_{r+3}+c_{r+2}^2+c_{r+1}^{2}+c_{r}^4=\begin{cases}
	1, & r =n-2 \\
	0, & r \in \{1,\dots,n-3\}.
	\end{cases}
	\label{eq:cond4}
	\end{equation}	
	We show by induction that the constraints we have obtained so far imply 
	\begin{equation}
		c_i=\begin{cases}
			c_0^{2^i}, & i \text{ odd } \\
			c_0^{2^i}+1, & i\text{ even }.
		\end{cases}
	\label{eq:ind}
	\end{equation}
	for all $i \in \{1,\dots,n-1\}$. The cases $i=1$ and $i=2$ are shown in Eq.~\eqref{eq:cond1} and \eqref{eq:cond3}. We verify the case $i=3$ by using Eq.~\eqref{eq:3case} with $r=1$:
	\begin{align*}
		0&=1+c_3^{2^{-2}}+c_2^{2^{-1}}+c_1^{2^{-1}}+c_0\\&=1+c_3^{2^{-2}}+c_0^2+1+c_0+c_0=c_3^{2^{-2}}+c_0^2,
	\end{align*}
	which immediately yields $c_3=c_0^8$ as claimed. We now proceed by induction: Assume that all coefficients up to $k\geq 3$ satisfy Eq.~\eqref{eq:ind}. Then by Eq. \eqref{eq:cond4} 
	\begin{equation*}
		c_{k+1}=c_k^2+c_{k-1}^2+c_{k-2}^4=(c_0^{2^{k+1}}+1)+c_0^{2^{k}}+(c_0^{2^{k}}+1)=c_0^{2^{k+1}}
	\end{equation*}
	if $k+1$ is odd and
	\begin{equation*}
		c_{k+1}=c_k^2+c_{k-1}^2+c_{k-2}^4=c_0^{2^{k+1}}+(c_0^{2^{k}}+1)+c_0^{2^{k}}=c_0^{2^{k+1}}+1
	\end{equation*}
	if $k+1$ is even, proving Eq.~\eqref{eq:ind}. 
	
	We compute $c_{n-1}$ in another way using Eq.~\eqref{eq:cond2} for $r=1$:
	\begin{equation} \label{eq:n-1}
		0=c_0+c_{n-1}^2+c_2^{2^{-1}}+c_1=c_{n-1}^2+c_0+c_0^2+1+c_0^2,
	\end{equation}
	so $c_{n-1}=1+c_0^{2^{n-1}}$. This immediately implies in connection with Eq.~\eqref{eq:ind} that $n$ must be odd. 
	
	The last step is to find a contradiction to the coefficients described in Eq.~\eqref{eq:ind}. For that, we use the condition from Eq. \eqref{eq:a2}. First, we expand the condition
	\begin{align*}
		Q&((x^2+x)\sum_{i=0}^{n-1} c_ix^{2^i}) =Q(\sum_{i=0}^{n-1} c_ix^{2^i+2}+\sum_{i=0}^{n-1} c_ix^{2^i+1})\\
		 =& \sum_{0\leq r<s\leq n-1} \left(\sum_{i=0}^{n-1} (c_{i-r+1}^{2^{r-1}}+c_{i-r}^{2^{r}})x^{2^i+2^{r}}\right)\\
		&\cdot\left(\sum_{j=0}^{n-1}(c_{j-s+1}^{2^{s-1}}+c_{j-s}^{2^{s}})x^{2^j+2^{s}}\right) \\
		=&\sum_{0\leq r<s\leq n-1}\sum_{i=0}^{n-1}\sum_{j=0}^{n-1} d_{i,j,r,s}x^{2^i+2^j+2^r+2^s}
	\end{align*}
	with $d_{i,j,r,s}=(c_{i-r+1}^{2^{r-1}}+c_{i-r}^{2^{r}})(c_{j-s+1}^{2^{s-1}}+c_{j-s}^{2^{s}})$. To satisfy condition~\eqref{eq:a2}, this polynomial must be equal to the zero polynomial. Using Eq.~\eqref{eq:ind} and Eq.~\eqref{eq:n-1}, we see that 
	\begin{equation*}
		d_{i,j,r,s}=\begin{cases} 0, & i=r \text{ or } j=s \\
			1, & \text{else}.
		\end{cases}
	\end{equation*}
	We check the coefficient of $x^8$ of the polynomial $Q((x^2+x)L_2^*(x))$: The possible choices for $i,j,r,s$ are: 
	\begin{enumerate}
		\item $i=j=0,r=1,s=2$ with $d_{i,j,r,s}=1$
		\item $i=r=0,s=1,j=2$ with $d_{i,j,r,s}=0$
		\item $i=r=0,j=1,s=2$ with $d_{i,j,r,s}=0$
		\item $j=r=0,i=1,s=2$ with $d_{i,j,r,s}=1$
		\item $j=r=0,s=1,i=2$ with $d_{i,j,r,s}=1$.
	\end{enumerate}
	In particular, the coefficient of $x^8$ is the sum of the listed values of $d_{i,j,r,s}$, which is $1$, so $Q((x^2+x)L_2(x))$ is not the zero polynomial. This contradicts condition~\eqref{eq:a2} and proves the theorem.

\end{proof}

\begin{remark}
	The condition $n\geq 5$ in Theorem~\ref{thm:noperm} is necessary. Indeed, it is possible to find permutation polynomials of the form $L_1(x^{-1})+L_2(x)$ over $\F_{2^4}$ and $\F_{2^3}$ using a simple computer search, examples with $L_1(x)=x$ can be found in \cite{eainverse}.
\end{remark}

Our main result is a direct consequence from Theorem~\ref{thm:noperm} and Proposition~\ref{prop:start} (recall that the inverse function is an involution).

\begin{theorem}[Main result] \label{thm:mainresult}
	Let $F \colon \F_{2^n} \rightarrow \F_{2^n}$ be the inverse function with $n\geq 5$. The CCZ-class of $F$ coincides with the EA-class of $F$. Moreover, all permutations in the CCZ-class of $F$ are affine equivalent to $F$.
\end{theorem}

\section{Conclusion and possible directions of future research}
We have shown that no permutation polynomials of the form $L_1(x^{-1})+L_2(x)$ in $\F_{2^n}$ for $n\geq 5$ exist. This implies that every function that is CCZ-equivalent to the inverse function is already EA equivalent to it, and that all permutations CCZ-equivalent to the inverse function are affine equivalent to it. An interesting avenue of further research is to consider the same questions for other functions with good cryptographic properties (nonlinearity/differential uniformity). In particular, Theorem~\ref{thm:mainresult} proves Conjecture~\ref{conj:bud} for the case of the inverse function. However, all other cases have not been answered yet (to our knowledge). Using the approach in this paper, a possible way to prove the conjecture would be to prove the non-existence of permutation polynomials of the form $L_1(x^d)+L_2(x)$ with $L_1,L_2\neq 0$ where $x^d$ is a non-Gold APN monomial. Note however, that the non-existence of such a polynomial is a stronger statement than the statement in Conjecture~\ref{conj:bud}, i.e. finding a permutation polynomial of the form $L_1(x^d)+L_2(x)$ with nonzero $L_1,L_2$ does not disprove Conjecture~\ref{conj:bud}.

More generally, an interesting way to expand on the results in this paper would be to work towards a classification of permutation polynomials of the form $L_1(x^d)+L_2(x)$ (or even just $x^d+L(x)$) over $\F_{2^n}$ where $L,L_1,L_2$ are linear mappings. Combined with Proposition~\ref{prop:start}, this might give insight into the CCZ-classes of monomials. Note that this family of permutation polynomials is similar to other families that have been investigated in the past, for instance the polynomials $x^s+\gamma \Tr(x^t)$ considered in \cite{pascalegohar}.

It would also be of interest to look at the same problem in odd characteristic. As shown in~\cite{oddchar}, there is no permutation polynomial of the form $L_1(x^{-1})+L_2(x)$ in characteristic $\geq 5$ which is a straightforward observation based on the fact that there is no Kloosterman zero in finite fields of characteristic $ \geq 5$. The characteristic $3$ case is however still open. 

\section*{Acknowledgment}
I would like to thank Gohar Kyureghyan and Faruk G\"olo\u{g}lu for many discussions, hints and encouragement throughout my work on this paper. I would also like to thank the anonymous referees for their comments that improved the presentation of this paper.

\bibliographystyle{IEEEtranS}

\bibliography{IEEEabrv,inverse}

\begin{IEEEbiographynophoto}{Lukas K\"olsch}
	received the M.Sc. degree from Otto von Guericke University, Magdeburg, in 2017, and the Ph.D. degree from University of Rostock in 2020. He is currently with University of Rostock, Germany. His research interests include cryptography, Boolean functions, and finite fields.
\end{IEEEbiographynophoto}
\end{document}